\documentclass[superscriptaddress]{revtex4}
\usepackage{amsmath, amsthm, amssymb}
\def\eps{\epsilon}
\newcommand{\bra}[1]{\l\langle #1\r|}
\newcommand{\ket}[1]{\l|#1\r\rangle}
\newcommand{\proj}[1]{|#1\rangle\langle #1|}
\newcommand{\braket}[2]{\l\langle #1|#2\r\rangle}

\DeclareMathOperator{\poly}{poly}

\def\benum#1\eenum{\begin{enumerate}#1\end{enumerate}}
\def\be#1\ee{\begin{equation}#1\end{equation}}
\def\bml#1\eml{\begin{multline}#1\end{multline}}
\def\ba#1\ea{\begin{align}#1\end{align}}
\def\bas#1\eas{\begin{align*}#1\end{align*}}
\def\eq#1{(\ref{eq:#1})}
\def\secref#1{Section~\ref{sec:#1}}
\def\thmref#1{Theorem~\ref{thm:#1}}
\def\lemref#1{Lemma~\ref{lem:#1}}
\def\bit{\begin{itemize}}
\def\eit{\end{itemize}}
\def\l{\left}
\def\r{\right}
\def\ot{\otimes}
\def\tl{\tilde{\lambda}}

\newtheorem{thm}{Theorem}

\newtheorem{lem}[thm]{Lemma}

\def\vb{\vec{b}}
\def\vx{\vec{x}}
\def\ra{\rightarrow}
\def\bbC{\mathbb{C}}
\def\bbZ{\mathbb{Z}}
\def\bbE{\mathbb{E}}
\newcommand{\Real}{\textrm{Re}}

\begin{document}


\title{Quantum algorithm for linear systems of equations}

\author{Aram W. Harrow}
\affiliation{Department of Mathematics, University
    of Bristol, Bristol, BS8 1TW, U.K.}
  \author{Avinatan
  Hassidim}
  \affiliation{MIT - Research Laboratory for Electronics,
    Cambridge, MA 02139, USA}
    \author{Seth Lloyd}\affiliation{MIT - Research
    Laboratory for Electronics and Department of Mechanical Engineering,
    Cambridge, MA 02139, USA}


\begin{abstract}

Solving linear systems of equations is
a common problem that arises both on its own and as a subroutine
in more complex problems: given a matrix ${A}$ and
a vector $\vec b$, find a vector
$\vec x$ such that ${A}\vec x = \vec b$. We consider the case where one doesn't need to know the solution $\vec x$ itself, but rather an
approximation of the expectation value of some operator associated
with $\vec x$, e.g., $\vec x^\dagger M \vec x$ for some matrix
${M}$.  In this case, when ${A}$ is sparse, $N\times N$ and has
condition number $\kappa$, classical
algorithms can find $\vec x$ and estimate $\vec x^\dag M \vec x$ in
$\tilde O(N\sqrt{\kappa})$ time.  Here, we exhibit a quantum algorithm
for this task that runs in $\poly(\log N, \kappa)$ time, an exponential
improvement over the best classical algorithm.
\end{abstract}

\maketitle

\section{Introduction}
Quantum computers are devices that harness quantum mechanics to
perform computations in ways that classical computers cannot.  For
certain problems, quantum algorithms supply exponential speedups over
their classical counterparts, the most famous example being Shor's
factoring algorithm \cite{Shor:1994:AQC}.  Few such exponential
speedups are known, and those that are (such as the use of quantum
computers to simulate other quantum systems \cite{Lloyd96a}) have so
far found limited use outside the domain of quantum mechanics.  This
paper presents a quantum algorithm to estimate features of the
solution of a set of linear equations.  Compared to classical
algorithms for the same task, our algorithm can be as much as
exponentially faster.

Linear equations play an important role in virtually all fields of
science and engineering.  The sizes of the data sets that define the
equations are growing rapidly over time, so that terabytes and even
petabytes of data may need to be processed to obtain a solution. In
other cases, such as when discretizing partial differential equations,
the linear equations may be implicitly defined and thus far larger
than the original description of the problem.  For a classical
computer even to approximate the solution of $N$ linear equations in
$N$ unknowns in general requires time that scales at least as $N$.
Indeed, merely to write out the solution takes time of order $N$.
Frequently, however, one is interested not in the full solution to the
equations, but rather in computing some function of that solution,
such as determining the total weight of some subset of the indices.
We show that in some cases, a quantum computer can approximate the
value of such a function in time which scales logarithmically in $N$,
and polynomially in the condition number (defined below) and desired
precision.  The dependence on $N$ is exponentially better than what is
achievable classically, while the dependence on condition number is
comparable, and the dependence on error is worse.  Thus our algorithm
can achieve useful, and even exponential, speedups in a wide variety
of settings where $N$ is large and the condition number is small.

We sketch here the basic idea of our algorithm, and then discuss it in
more detail in the next section.  Given a Hermitian $N\times N$ matrix
$A$, and a unit vector $\vec b$, suppose we would like to find $\vec
x$ satisfying $A\vec x=\vec b$.  (We discuss later questions of
efficiency as well as how the assumptions we have made about $A$ and
$\vec b$ can be relaxed.)  First, the algorithm represents $\vec b$ as
a quantum state $\ket b = \sum_{i=1}^N b_i \ket i$.  Next, we use
techniques of Hamiltonian simulation\cite{berry2007eqa,childs-2008} to
apply $e^{iAt}$ to $\ket b$ for a superposition of different times
$t$.  This ability to exponentiate $A$ translates, via the well-known
technique of phase estimation\cite{LP96,cleve1997qar,BDM98}, into the
ability to decompose $\ket{b}$ in the eigenbasis of $A$ and to find
the corresponding eigenvalues $\lambda_j$. Informally, the state of
the system after this stage is close to $\sum_{j=1}^N \beta_j
\ket{u_j}\ket{\lambda_j}$, where $u_j$ is the eigenvector basis of
$A$, and $\ket{b} = \sum_{j=1}^N \beta_j \ket{u_j}$.  We would then
like to perform the linear map taking $\ket{\lambda_j}$ to
$C\lambda_j^{-1} \ket{\lambda_j}$, where $C$ is a normalizing
constant.  As this operation is not unitary, it has some probability
of failing, which will enter into our discussion of the run-time
below.  After it succeeds, we uncompute the $\ket{\lambda_j}$ register
and are left with a state proportional to $\sum_{j=1}^N \beta_j
\lambda_j^{-1} \ket{u_j} = A^{-1}\ket b = \ket x$.

An important factor in the performance of the matrix inversion
algorithm is $\kappa$, the condition number of $A$, or the ratio
between $A$'s largest and smallest eigenvalues. As the condition
number grows, $A$ becomes closer to a matrix which cannot be inverted,
and the solutions become less stable.  Such a matrix is said to be
``ill-conditioned.''  Our algorithms will generally assume that the
singular values of $A$ lie between $1/\kappa$ and 1; equivalently
$\kappa^{-2}I \leq A^\dag A \leq I$.  In this case, our runtime will
scale as $\kappa^2\log(N)/\eps$, where $\eps$ is the additive error
achieved in the output state $\ket x$.  Therefore, the greatest
advantage our algorithm has over classical algorithms occurs when both
$\kappa$ and 
$1/\eps$ are $\poly\log(N)$, in which case it achieves an exponential
speedup.  However, we will also discuss later some 
techniques for  handling ill-conditioned matrices.

This procedure yields a quantum-mechanical representation $\ket{x}$ of
the desired vector $\vec x$.  Clearly, to read out all the components
of $\vec x$ would require one to perform the procedure at least $N$
times.  However, often one is interested not in $\vec x$ itself, but in
some expectation value $\vec x^T M \vec x$, where $M$ is some linear
operator (our procedure also accommodates nonlinear operators as
described below).  By mapping $M$ to a quantum-mechanical operator,
and performing the quantum measurement corresponding to $M$, we obtain
an estimate of the expectation value $\bra{x} M \ket{x} = \vec x^T M
\vec x$, as desired.  A wide variety of features of the vector $\vec
x$ can be extracted in this way, including normalization, weights in
different parts of the state space, moments, etc.

A simple example where the algorithm can be used is to see if two
different stochastic processes have similar stable state
\cite{luen79}. Consider a stochastic process $\vec x_t = A \vec
x_{t-1} + \vec b$,
where the $i$'th coordinate in the vector $\vec x_t$ represents the
abundance of item $i$ in time $t$. The stable state of this
distribution is given by $\ket{x} = (I - A)^{-1}\ket{b}$. Let $\vec
x'_t = A' \vec x'_{t-1} +  \vec b'$, and $\ket{x'} = (I - 
A')^{-1}\ket{b'}$. To know if $\ket{x}$ and $\ket{x'}$ are similar, we
perform the SWAP test between them 
\cite{buhrman2001qf}. We note that classically finding out if two
probability distributions are similar requires at least $O(\sqrt{N})$
samples \cite{Valiant08}.

The strength of the algorithm is that it works only with $O(\log
N)$-qubit registers, and never has to write down all of $A$, $\vec b$
or $\vec x$.  In situations (detailed below) where the Hamiltonian
simulation and our non-unitary step incur only $\poly\log(N)$
overhead, this means our algorithm takes exponentially less time than
a classical computer would need even to write down the output.  In
that sense, our algorithm is related to classical Monte Carlo
algorithms, which achieve dramatic speedups by working with samples
from a probability distribution on $N$ objects rather than by writing
down all $N$ components of the distribution.  However, while these
classical sampling algorithms are powerful, we will prove that in fact
{\em any} classical algorithm requires in general exponentially more
time than our quantum algorithms to perform the same matrix inversion
task.

\paragraph*{Outline}
The rest of the Letter proceeds by first describing our algorithm in
detail, analyzing its run-time and comparing it with the best known
classical algorithms.  Next, we prove (modulo some
complexity-theoretic assumptions) hardness results for matrix
inversion that imply both that our algorithm's run-time is nearly
optimal, and that it runs exponentially faster than any classical
algorithm.  We conclude with a discussion of applications,
generalizations and extensions.

\paragraph*{Related work}
Previous papers gave quantum algorithms to perform linear algebraic
operations in a limited setting \cite{klappenecker2003qpt}.  Our work
was extended by \cite{LO-08} to solving nonlinear differential
equations.

\section{Algorithm}
We now give a more detailed explanation of the algorithm.  First, we
want to transform a given Hermitian matrix $A$ into a unitary operator
$e^{iA t}$ which we can apply at will.  This is possible (for example)
if $A$ is $s$-sparse and efficiently row computable, meaning it has at
most $s$ nonzero entries per row and given a row index these entries
can be computed in time $O(s)$.  Under these assumptions,
Ref.~\cite{berry2007eqa} shows how to simulate $e^{iAt}$ in time 
\[\tilde{O}(\log(N) s^2 t), \] where the $\tilde O$ suppresses more
slowly-growing terms (described in \cite{HHL-08}). If $A$ is not
Hermitian, define
\begin{equation} C = \begin{pmatrix}
  0     & A \\
  A^{\dagger} & 0
\end{pmatrix}
\end{equation}
As $C$ is Hermitian, we can solve the equation $C\vec y
= \begin{pmatrix}\vec b \\ 0 \end{pmatrix}$ to obtain $y
= \begin{pmatrix}0 \\ \vec x\end{pmatrix}$. Applying this reduction if
necessary, the rest of the Letter assumes that $A$ is Hermitian. 

We also need an efficient procedure to prepare $\ket{b}$.  For
example, if $b_i$ and $\sum_{i=i_1}^{i_2} |b_i|^2$ are efficiently
computable then we can use the procedure of Ref.~\cite{grover02} to
prepare $\ket{b}$.  Alternatively, our algorithm could be a subroutine
in a larger quantum algorithm of which some other component is responsible for
producing $\ket b$.

The next step is to decompose $\ket{b}$ in the eigenvector
basis, using phase estimation \cite{LP96,cleve1997qar,BDM98}.  Denote by
$\ket{u_j}$ the eigenvectors of $A$ (or equivalently, of $e^{iAt}$),
and by $\lambda_j$ the corresponding eigenvalues. Let 
\be \ket{\Psi_0} :=
\sqrt{\frac{2}{T}} \sum_{\tau=0}^{T-1}
\sin\frac{\pi(\tau+\frac{1}{2})}{T} \ket{\tau}
\ee
for some large $T$.  The coefficients of $\ket{\Psi_0}$ are chosen
(following \cite{LP96,BDM98}) to minimize a certain quadratic loss function
which appears in our error analysis (see \cite{HHL-08} for
details).

Next we apply the conditional Hamiltonian evolution
$\sum_{\tau=0}^{T-1} \proj{\tau}^C \ot e^{iA\tau t_0 / T}$ on
$\ket{\Psi_0}^C \otimes \ket{b}$, where  $t_0 =
O(\kappa/\epsilon)$.
Fourier transforming the first register gives the state
\begin{equation}
 \sum_{j=1}^N \sum_{k=0}^{T-1} \alpha_{k|j} \beta_j \ket{k}\ket{u_j},
\label{eq:alpha-beta-state}\end{equation}
\noindent where $\ket{k}$ are the Fourier basis states, and $|\alpha_{k|j}|$ is large if and only if $\lambda_j
\approx \frac{2 \pi k}{t_0}$. Defining $\tilde{\lambda}_k := 2\pi k
/t_0$, we can relabel our $\ket{k}$ register to obtain
\[
\sum_{j=1}^N \sum_{k=0}^{T-1}  \alpha_{k|j} \beta_j \ket{\tilde{\lambda}_k}\ket{u_j}\]
Adding an ancilla qubit and rotating conditioned on $\ket{\tl_k}$ yields
\[
\sum_{j=1}^N \sum_{k=0}^{T-1}  \alpha_{k|j} \beta_j
\ket{\tilde{\lambda}_k}\ket{u_j} \left( \sqrt{1 -
    \frac{C^2}{\tilde{\lambda}_k^2}} \ket{0} +
  \frac{C}{\tilde{\lambda}_k}\ket{1} \right),\] where $C  =
O(1/\kappa)$. We now undo the phase estimation to uncompute the
$\ket{\tilde{\lambda}_k}$.  If the phase estimation were perfect, we
would have  $\alpha_{k|j} = 1$ if $\tilde{\lambda}_k = \lambda_j$, and
$0$ otherwise. Assuming this for now, we obtain
\[
\sum_{j=1}^N \beta_j \ket{u_j} \left( \sqrt{1 - \frac{C^2}{\lambda_j^2}} \ket{0} + \frac{C}{\lambda_j}\ket{1} \right)\]

To finish the inversion we measure the last qubit. Conditioned on seeing 1, we have the state
\[ \sqrt{\frac{1}{\sum_{j=1}^N C^2|\beta_j|^2 / |\lambda_j|^2}} \sum_{j=1}^N
\beta_j \frac{C}{\lambda_j} \ket{u_j} \] which corresponds to $\ket{x}
= \sum_{j=1}^n \beta_j\lambda_j^{-1} \ket{u_j}$ up to
normalization. We can determine the normalization factor from the
probability of obtaining 1. Finally, we make a measurement $M$ whose
expectation value $\bra{x} M \ket{x}$ corresponds to the feature of
$\vec x$ that we wish to evaluate.

\paragraph*{Run-time and error analysis}
We present an informal description of the sources of error; the exact
error analysis and runtime considerations are presented in
\cite{HHL-08}. Performing the phase estimation is done by simulating
$e^{iAt}$.  Assuming that $A$ is $s$-sparse, this can be done with
 error $\eps$ in time proportional to $ts^2(t/\eps)^{o(1)} =: \tilde O(ts^2)$.

The dominant source of error is phase estimation.  This step errs by
$O(1/t_0)$ in estimating $\lambda$, which translates into a relative
error of $O(1/\lambda t_0)$ in $\lambda^{-1}$.  If $\lambda \geq
1/\kappa$ taking $t_0 = O(\kappa/\eps)$ induces a final error of
$\eps$.  Finally, we consider the success probability of the
post-selection process.  Since $C=O(1/\kappa)$ and $\lambda \leq 1$,
this probability is at least $\Omega(1/\kappa^2)$.  Using amplitude
amplification \cite{BHMT02}, we find that $O(\kappa)$ repetitions are
sufficient. Putting this all together, 
we obtain the stated runtime of $\tilde{O} \left( \log(N)s^2\kappa^2 /
  {\eps} \right)$.

\section{Optimality}
\paragraph*{Classical matrix inversion algorithms}
To put our algorithm in context, one of the best general-purpose
classical matrix 
inversion algorithms is the conjugate gradient
method \cite{shewchuk94}, which, when $A$ is positive definite, uses
$O(\sqrt{\kappa}\log(1/\eps))$ matrix-vector multiplications each
taking time $O(Ns)$ for a total runtime of
$O(Ns\sqrt{\kappa}\log(1/\eps))$.  (If $A$ is not positive
definite, $O(\kappa\log(1/\eps))$ multiplications are required,
for a total time of $O(Ns{\kappa}\log(1/\eps))$.)
 An important question is whether
classical methods can be improved when only a summary statistic of the
solution, such as $\vec x^\dag M \vec x$, is required.  Another
question is whether our quantum algorithm could be improved, say to
achieve error $\eps$ in time proportional to $\poly\log(1/\eps)$.  We
show that the answer to both questions is negative, using an argument
from complexity theory.  Our strategy is to prove that the ability to
invert matrices (with the right choice of parameters) can be used to
simulate a general quantum computation.

\paragraph*{The complexity of matrix inversion}
We show that a quantum circuit using $n$ qubits and $T$ gates can
be simulated by inverting an $O(1)$-sparse matrix $A$
of dimension $N=O(2^n \kappa)$.  The
condition number $\kappa$ is $O(T^2)$ if we need $A$ to be
positive definite or $O(T)$ if not.
This implies that a classical $\poly(\log N, \kappa, 1/\eps)$-time
algorithm would be able to simulate a $\poly(n)$-gate quantum
algorithm in $\poly(n)$ time.  Such a simulation is strongly
conjectured to be false, and is known to be impossible in the presence
of oracles \cite{simon}.

The reduction from a general quantum circuit to a matrix inversion
problem also implies that our algorithm cannot be substantially
improved (under standard assumptions).  If the run-time could be made
polylogarithmic in $\kappa$, then any problem solvable on $n$ qubits
could be solved in $\poly(n)$ time (i.e. {\bf BQP}={\bf PSPACE}), a
highly unlikely possibility. 
Even improving our $\kappa$-dependence to $\kappa^{1-\delta}$ for
$\delta>0$ would allow any time-$T$ quantum algorithm to be simulated
in time $o(T)$; iterating this would again imply that {\bf BQP}={\bf
  PSPACE}.  Similarly, improving the error dependence to
$\poly\log(1/\eps)$ would imply that {\bf BQP} includes {\bf PP}, and
even minor improvements would contradict oracle lower bounds
\cite{parity}.

\paragraph*{The reduction}
We now present the key reduction from simulating a quantum circuit to
matrix inversion.  Let $\cal{C}$ be a quantum circuit acting on $n
= \log N$ qubits which applies $T$ two-qubit gates
$U_1, \ldots U_T$.  The initial state is $\ket{0}^{\otimes n}$ and the
answer is determined by measuring the first qubit of the final state.

Now adjoin an ancilla register of dimension $3T$ and
define a unitary
\be U = \sum_{t=1}^{T} \ket{t\!+\!1}\!\bra{t} \otimes U_t
\,+\, \ket{t\!+\!T\!+\!1}\!\bra{t\!+\!T} \otimes I
\\ 
\,+\, \ket{t\!+\!2T\!+\!1 \bmod{3T}}\!\bra{t\!+\!2T} \otimes U_{3T+1-t}^\dag.
\label{eq:feynman-unitary}\ee
We have chosen $U$ so that for $T+1\leq t\leq 2T$, applying $U^t$ to
$\ket{1}\ket{\psi}$ yields
$\ket{t+1}\otimes U_T\cdots U_1\ket{\psi}$.
If we now define $A = I - U e^{-1/T}$ then $\kappa(A) = O(T)$, and we can expand
\be A^{-1} = \sum_{k\geq 0} U^k e^{-k/T},
\label{eq:inv-expansion}\ee
This can be interpreted as applying $U^t$ for $t$ a
geometrically-distributed random variable.  Since $U^{3T}=I$, we can
assume $1\leq t \leq 3T$.  If we measure the first register and obtain
$T+1\leq t\leq 2T$ (which occurs
with probability $e^{-2}/(1+e^{-2}+e^{-4})\geq 1/10$) then we are left
with the second register in the state $U_T\cdots U_1\ket{\psi}$,
corresponding to a successful computation.  Sampling from $\ket{x}$
allows us to sample from the results of the computation.  This
establishes that matrix inversion is
{\bf BQP}-complete, and proves our above claims about the difficulty
of improving our algorithm.

\section{Discussion}
There are a number of ways to extend our algorithm and relax the assumptions
we made while presenting it.  We will discuss first how to invert a
broader class of matrices and then consider measuring other features of
$\vec x$ and performing operations on $A$ other than inversion.

Certain non-sparse $A$ can be simulated and therefore inverted;
see \cite{childs-2008} for techniques and examples. It is also
possible to invert non-square matrices, using the reduction presented
from the non-Hermitian case to the Hermitian one.

The matrix inversion algorithm can also handle ill-conditioned
matrices by inverting only the part of $\ket{b}$ which is in the
well-conditioned part of the matrix. Formally, instead of transforming
$\ket{b} = \sum_j \beta_j \ket{u_j}$ to $\ket{x} = \sum_j
\lambda_j^{-1} \beta_j \ket{u_j}$, we transform it to a state which is
close to
\[ \sum_{j, \lambda_j < 1/\kappa}  \lambda_j^{-1} \beta_j \ket{u_j}
\ket{\text{well}}
+ \sum_{j, \lambda_j \ge 1/\kappa}  \beta_j \ket{u_j} \ket{\text{ill}}\]
in time proportional to $\kappa^2$ for any chosen $\kappa$ (i.e. not
necessarily the true condition number of $A$). The last qubit is a
flag which enables the user to estimate what the size of the
ill-conditioned part, or to handle it in any other way she wants. This
behavior can be advantageous if we know that $A$ is not 
invertible and we are interested in  the projection of $\ket{b}$ on
the well-conditioned part of $A$.

Another method that is often used in classical algorithms to handle
ill-conditioned matrices is to apply a preconditioner\cite{Chen05}.
If we have a method of generating a preconditioner matrix $B$ such
that $\kappa(AB)$ is smaller than $\kappa(A)$, then we can solve
$A\vec x=\vec b$ by instead solving the possibly easier matrix
inversion problem $(AB)\vec c = B\vec b$.  Further, if $A$ and $B$ are
both sparse, then $AB$ is as well.  Thus, as long as a state
proportional to $B\ket{b}$ can be efficiently prepared, our algorithm
could potentially run much faster if a suitable preconditioner is
used.

The outputs of the algorithm can also be generalized. We can estimate
degree-$2k$ polynomials in the entries of $\vec x$ by generating $k$
copies of $\ket{x}$ and measuring the $nk$-qubit observable
$$\sum_{i_1,\ldots,i_k,j_1,\ldots,j_k} M_{i_1,\ldots,i_k,j_1,\ldots,j_k} \ket{i_1,\ldots,i_k}\bra{j_1,\ldots,j_k}$$
on the state $\ket{x}^{\otimes k}.$ Alternatively, one can use our
algorithm to generate a quantum analogue of Monte-Carlo, where given
$A$ and $\vb$ we sample from the vector $\vx$, meaning that the value $i$
occurs with probability $|\vec{x}_i|^2 $.

Perhaps the most far-reaching generalization of the matrix inversion
algorithm is not to invert matrices at all! Instead, it can compute
$f(A)\ket{b}$ for any computable $f$. Depending on the degree of
nonlinearity of $f$, nontrivial tradeoffs between accuracy and
efficiency arise.  Some variants of this idea are considered in
\cite{SMM08,childs-2008,LO-08}.

{\bf Acknowledgements.}
We thank the W.M. Keck foundation for support, and
AWH thanks them as well as MIT for hospitality while this work was
carried out.  AWH was also funded by the U.K. EPSRC grant ``QIP IRC.'' SL
thanks R. Zecchina for encouraging him to work on this problem.
We are grateful as well to R. Cleve, D. Farmer, S. Gharabian, J.
Kelner, S. Mitter, P. Parillo, D. Spielman and M. Tegmark for helpful
discussions.

\appendix
\section{Proof details}

In this appendix, we describe and analyze our algorithm in full
detail.  While the body of the paper attempted to convey the spirit of
the procedure and left out various improvements, here we take the
opposite approach and describe everything, albeit possibly in a less
intuitive way.  We also describe in more detail our reductions from
non-Hermitian matrix inversion to Hermitian matrix inversion
(\secref{non-hermitian}) 
and from a general quantum computation to matrix inversion
(\secref{reduction}).

As inputs we require a procedure to produce the state $\ket{b}$, a
method of producing the $\leq s$ non-zero elements of any row of $A$
and a choice of cutoff $\kappa$.  Our run-time will be roughly
quadratic in $\kappa$ and our algorithm is guaranteed to be correct if
$\|A\| \leq 1$ and $\|A^{-1}\| \leq \kappa$.

The condition number is a crucial parameter in the algorithm.  Here we
present one possible method of handling
ill-conditioned matrices.  We will define the well-conditioned part of
$A$ to be the span of the eigenspaces corresponding to eigenvalues
$\geq 1/\kappa$ and the ill-conditioned part to be the rest.  Our
strategy will be to flag the ill-conditioned part of the
matrix (without inverting it), and let the user choose how to
further handle this.   Since we cannot exactly resolve any eigenvalue,
we can only approximately determine whether vectors are in the well-
or ill-conditioned subspaces.   Accordingly, we choose some $\kappa' >
\kappa$ (say $\kappa'=2\kappa$).  Our algorithm then inverts the
well-conditioned part of the matrix,
 flags any eigenvector with eigenvalue $\leq 1/\kappa'$ as
 ill-conditioned, and interpolates between these two behaviors when
 $1/\kappa' < |\lambda| < 1/\kappa$.  This is described formally in the
 next section.
 We present this strategy not because it is necessarily ideal in all
 cases, but because it gives a concrete illustration of the key
 components of our algorithm.

Finally, the algorithm produces $\ket{x}$ only up to some error $\eps$
which is given as part of the input.  We work only with pure states, and so
define error in terms of distance between vectors, i.e. $\|\,
\ket{\alpha}  - \ket{\beta}\,\| = \sqrt{2(1-\Real
  \braket{\alpha}{\beta})}$.  Since ancilla states are produced and
then imperfectly uncomputed by the algorithm, our output state will
technically have high fidelity not with $\ket{x}$ but with
$\ket{x}\ket{000\ldots}$.  In general we do not write down ancilla
qubits in the $\ket{0}$ state, so we write $\ket{x}$ instead of
$\ket{x}\ket{000\ldots}$ for the target state, $\ket{b}$ instead of
$\ket{b}\ket{000\ldots}$ for the initial state, and so on.

\subsection{Detailed description of the algorithm}

To produce the input state $\ket{b}$, we  assume that there exists an
efficiently-implementable
unitary $B$, which when applied
to $\ket{\text{initial}}$ produces the state $\ket{b}$, possibly
along with garbage in an ancilla register.  We make no further assumption
about $B$; it may represent another part of a larger algorithm, or a
standard state-preparation procedure such as \cite{grover02}.  Let $T_B$
be the number of gates required to implement $B$.  We neglect the
possibility that $B$ errs in producing $\ket{b}$ since, without any
other way of producing or verifying the state $\ket{b}$, we have no
way to mitigate these errors. Thus, any errors in producing $\ket{b}$
necessarily translate directly into errors in the final state $\ket{x}$.

Next, we define the state
\be \ket{\Psi_0} =
\sqrt{\frac{2}{T}} \sum_{\tau=0}^{T-1}
\sin\frac{\pi(\tau+\frac{1}{2})}{T} \ket{\tau}
\label{eq:buzek-state}\ee
for a $T$ to be chosen later.  Using \cite{grover02}, we can prepare
$\ket{\Psi_0}$ up to error $\eps_{\Psi}$ in time
$\poly\log(T/\eps_\Psi)$.

One other subroutine we will need is Hamiltonian simulation.  Using
the reductions described in \secref{non-hermitian}, we can assume that
$A$ is Hermitian.
To simuluate $e^{iAt}$ for some $t\geq 0$, we use the algorithm of
\cite{berry2007eqa}.  If $A$ is $s$-sparse, $t\leq t_0$ and we want to
guarantee that the error is $\leq \eps_H$, then this requires time 
\be
T_H = O(\log(N)(\log^*(N))^2s^2t_0 9^{\sqrt{\log(s^2t_0/\eps_H)}})
 = \tilde O (\log(N)s^2 t_0)
\label{eq:precise-H-sim-time}\ee
The scaling here is better than any power of $1/\eps_H$, which means
that the additional error introduced by this step introduces is
negligible compared with the rest of the algorithm, and the runtime is
almost linear with $t_0$. Note that this is the only step where we
require that $A$ be sparse; as there are some other types of
Hamiltonians which can be simulated efficiently
(e.g. \cite{aharonov2003aqs,berry2007eqa,childs-2008}), this broadens
the
set of matrices we can handle.

The key subroutine of the algorithm, denoted $U_{\text{invert}}$, is defined as
follows:
\begin{enumerate}
\item Prepare $\ket{\Psi_0}^C$ from $\ket{0}$ up to error $\eps_\Psi$.
\item Apply the conditional Hamiltonian evolution $\sum_{\tau=0}^{T-1}
  \proj{\tau}^C \ot e^{iA\tau t_0 / T}$ up to error $\eps_H$.
\item Apply the Fourier transform to the register $C$.  Denote the
  resulting basis states with $\ket{k}$, for $k=0,\ldots T-1$.  Define
  $\tilde{\lambda}_k := 2\pi k /t_0$.
\item Adjoin a three-dimensional register $S$ in the state
$$\ket{h(\tl_k)}^S :=
\sqrt{1-f(\tl_k)^2 - g(\tl_k)^2}\ket{\text{nothing}}^S
+ f(\tl_k)\ket{\text{well}}^S
+ g(\tl_k)\ket{\text{ill}}^S,$$
for functions $f(\lambda),g(\lambda)$ defined below in \eq{fg-defs}.  Here
`nothing' indicates
that the desired matrix inversion hasn't taken place, `well'
indicates that it has, and `ill' means that part of $\ket{b}$ is in
the ill-conditioned subspace of $A$.
\item Reverse steps 1-3, uncomputing any garbage produced
  along the way.
\end{enumerate}

The functions $f(\lambda),g(\lambda)$ are known as filter
functions\cite{Hansen98}, and are chosen so that for some
constant $C>1$:
$f(\lambda)=1/C\kappa\lambda$ for $\lambda \geq 1/\kappa$,
$g(\lambda)=1/C$ for $\lambda \leq 1/\kappa':= 1/2\kappa$ and $f^2(\lambda) +
g^2(\lambda)\leq 1$ for all $\lambda$.  Additionally, $f(\lambda)$
should satisfy a certain continuity property that we will describe in
the next section.  Otherwise the functions are arbitrary.  One
possible choice is
\begin{subequations}\label{eq:fg-defs}
\be f(\lambda) =
   \left\{ \begin{array}{ll}
       \frac{1}{2\kappa\lambda} &
\text{when } \lambda \ge 1/\kappa \\
\frac{1}{2}
 \sin \left( \frac{\pi}{2}\cdot \frac{\lambda -
     \frac{1}{\kappa^{\prime}}} {\frac{1}{\kappa} - \frac{1}{\kappa'}}
 \right) & \text{when }
\frac{1}{\kappa} > \lambda \ge \frac{1}{\kappa^{\prime}}\\
       0 &\text{when }
\frac{1}{\kappa^{\prime}} > \lambda \end{array} \right. \ee

\be g(\lambda) =
   \left\{ \begin{array}{ll}
       0 &\text{when } \lambda \ge 1/\kappa \\
\frac{1}{2} \cos \left( \frac{\pi}{2}\cdot \frac{\lambda -
     \frac{1}{\kappa^{\prime}}} {\frac{1}{\kappa} - \frac{1}{\kappa'}}
 \right) &\text{when } \frac{1}{\kappa} > \lambda \ge \frac{1}{\kappa^{\prime}}\\
\frac{1}{2}        &\text{when }
\frac{1}{\kappa^{\prime}} > \lambda \end{array} \right. \ee
\end{subequations}
If $U_{\text{invert}}$ is applied to $\ket{u_j}$ it will, up to an
error we will discuss below, adjoin the state $\ket{h(\lambda_j)}$.  Instead
if we apply $U_{\text{invert}}$ to $\ket{b}$ (i.e. a superposition of
different $\ket{u_j}$), measure $S$ and obtain the outcome `well',
then we will have approximately applied an operator proportional to
$A^{-1}$.  Let $\tilde p$ (computed in the next section) denote the success
probability of this measurement.  Rather than repeating $1/\tilde p$
times, we will use amplitude amplification \cite{BHMT02} to obtain the
same results with $O(1/\sqrt{\tilde p})$ repetitions.  To describe the
procedure, we introduce two new operators:
$$R_{\text{succ}} = I^S- 2\proj{\text{well}}^S,$$
acting only on the $S$ register and
$$R_{\text{init}} = I - 2\proj{\text{initial}}.$$

Our main algorithm then follows the amplitude amplification procedure:
we start with $U_{\text{invert}}B\ket{\text{initial}}$ and repeatedly apply
$U_{\text{invert}} B R_{\text{init}} B^\dag U_{\text{invert}}^\dag
R_{\text{succ}}$.  Finally we measure $S$ and stop when we obtain the
result `well'.  The number of repetitions would ideally be
$\pi/4\sqrt{\tilde p}$, which in the next section we will show is
$O({\kappa})$.   While $\tilde{p}$ is initially unknown, the
procedure has a constant probability of success if the number of
repetitions is a constant fraction of $\pi/4\tilde{p}$.  Thus, following
\cite{BHMT02} we repeat the entire procedure with a geometrically
increasing number of repetitions each time: 1, 2, 4, 8, \ldots, until we
have reached a power of two that is
$\geq {\kappa}$.  This yields a constant probability of success
using $\leq 4{\kappa}$ repetitions.

Putting everything together, the run-time is
$\tilde O({\kappa}(T_B + t_0s^2 \log(N))$, where the $\tilde O$
suppresses the more-slowly growing terms of $(\log^*(N))^2$,
$\exp(O(1/\sqrt{\log(t_0/\eps_H)}))$ and $\poly\log(T/\eps_\Psi)$.  In
the next section, we will show that $t_0$ can be taken to be
$O(\kappa/\eps)$ so that the total run-time is
$\tilde O({\kappa}T_B + \kappa^2
s^2 \log(N)/\eps)$.

\subsection{Error Analysis}
In this section we show that taking $t_0 = O(\kappa/{\eps})$
introduces an error of $\leq \eps$ in the final state.  The main
subtlety in analyzing the error comes from the post-selection step, in
which we choose only the part of the state attached to the
$\ket{\text{well}}$ register.  This can potentially magnify errors
in the overall state.  On the other hand, we may also be interested in
the non-postselected state, which results from applying
$U_{\text{invert}}$ a single time to $\ket{b}$.  For instance, this
could be used to estimate the amount of weight of $\ket{b}$ lying in
the ill-conditioned
components of $A$.  Somewhat surprisingly,
we show that the error in both cases is upper-bounded by
$O(\kappa/t_0)$.

In this section, it will be convenient to ignore the error terms
$\eps_H$ and $\eps_\Psi$, as these can be made negligible with
relatively little effort and it is the errors from phase estimation
that will dominate.  Let $\tilde{U}$ denote a version of
$U_{\text{invert}}$ in which everything except the phase estimation is
exact.  Since $ \|\tilde U - U_{\text{invert}} \| \leq O(\eps_H +
\eps_\Psi)$, it is sufficient to work with $\tilde U$.
Define $U$ to
be the ideal version of $U_{\text{invert}}$ in which there is no error
in any step.

\begin{thm}[Error bounds]\label{thm:error}\
\begin{enumerate}
\item In the case when no post-selection is performed, the error is bounded as
\be\|\tilde U - U \| \leq O(\kappa/t_0).\label{eq:no-ps-error}\ee
\item If we post-select on the flag register being in the space spanned by $\{\ket{\text{well}}, \ket{\text{ill}}\}$ and define the normalized ideal state to be $\ket{x}$ and our actual state to be $\ket{\tilde x}$ then
\be\|\,\ket{\tilde x}-\ket{x}\,\|\leq O(\kappa/t_0).\label{eq:fg-ps-error}\ee
\item If $\ket{b}$ is entirely within the well-conditioned subspace of $A$ and we post-select on the flag register being $\ket{\text{well}}$ then
\be\|\,\ket{\tilde x}-\ket{x}\,\|\leq O(\kappa/t_0). \label{eq:f-ps-error}\ee
\end{enumerate}
\end{thm}

The third claim is often of the most practical interest, but the other two are useful if we want to work with the ill-conditioned space, or estimate its weight.

The rest of the section is devoted to the proof of \thmref{error}. We first show that the third claim is a corollary of the second, and then prove the first two claims more or less independently. To prove (\ref{eq:fg-ps-error} assuming (\ref{eq:no-ps-error}), observe that if  $\ket{b}$ is entirely in the well-conditioned space, the ideal state $\ket{x}$ is proportional to
$A^{-1}\ket{b}\ket{\text{well}}$.  Model the post-selection on
$\ket{\text{well}}$ by a post-selection first on the space spanned by
$\{\ket{\text{well}}, \ket{\text{ill}}\}$, followed by a
post-selection onto
$\ket{\text{well}}$.  By (\ref{eq:no-ps-error}), the first
post-selection leaves us with error $O(\kappa/t_0)$.  This implies
that the second post-selection will succeed with probability $\geq 1 -
O(\kappa^2/t_0^2)$ and therefore will increase the error by at most
$O(\kappa/t_0)$.  The final error is then $O(\kappa/t_0)$ as claimed
in \eq{f-ps-error}.

Now we turn to the proof of \eq{no-ps-error}.
A crucial piece of the proof will be the following statement about the
continuity of $\ket{h(\lambda)}$.

\begin{lem}\label{lem:no-PS}
The map $\lambda\mapsto\ket{h(\lambda)}$ is $O(\kappa)$-Lipschitz,
meaning that for any $\lambda_1\neq\lambda_2$,
$$\| \, \ket{h(\lambda_1)}  - \ket{h(\lambda_2)}\, \|
 = \sqrt{2(1-\Real\braket{h(\lambda_1)}{h(\lambda_2)})}
\leq c\kappa |\lambda_1 - \lambda_2|,$$
for some $c= O(1)$.
\end{lem}

\begin{proof}
Since $\lambda\mapsto\ket{h(\lambda)}$ is continuous everywhere and
differentiable everywhere except at $1/\kappa$ and $1/\kappa'$, it
suffices to bound the norm of the derivative of $\ket{h(\lambda)}$.
We consider it piece by piece.  When $\lambda>1/\kappa$,
$$\frac{d}{d\lambda}\ket{h(\lambda)} =
\frac{1}{2\kappa^2\lambda^3\sqrt{1-1/2\kappa^2\lambda^2}}
\ket{\text{nothing}}
- \frac{1}{2\kappa\lambda^2}\ket{\text{well}},$$
which has squared norm
$\frac{1}{2\kappa^2\lambda^4(2\kappa^2\lambda^2-1)} +
\frac{1}{4\kappa^2\lambda^4} \leq \kappa^2$.
Next, when $1/\kappa'<\lambda < 1/\kappa$, the norm of
$\frac{d}{d\lambda}\ket{h(\lambda)}$ is
$$\frac{1}{2}\cdot\frac{\pi}{2}\cdot\frac{1}{\frac{1}{\kappa}-
\frac{1}{\kappa'}} = \frac{\pi}{2}\kappa.$$
Finally $\frac{d}{d\lambda}\ket{h(\lambda)}=0$ when $\lambda <
1/\kappa'$.  This completes the proof, with $c =
\frac{\pi}{2}$.
\end{proof}

Now we return to the proof of \eq{no-ps-error}.
Let $\tilde P$ denote the first three steps of the algorithm.  They
can be thought of as mapping the initial zero qubits to a
$\ket{k}$ register,
together with some garbage, as follows:
$$\tilde P = \sum_{j=1}^n \proj{u_j} \ot \sum_k \alpha_{k|j} \ket{k}
\ket{\text{garbage}(j,k)}\bra{\text{initial}},$$
where the guarantee that the phase estimation algorithm gives us is
that $\alpha_{k|j}$ is concentrated around $\lambda_j \approx 2\pi
k/t_0 =: \tl_k$.  Technically, $\tilde P$ should be completed to make
it a unitary operator by defining some arbitrary behavior on inputs
other than $\ket{\text{initial}}$ in the last register.

Consider a test state $\ket{b} = \sum_{j=1}^N \beta_j \ket{u_j}$.
  The
 ideal functionality is defined by
\[\ket{\varphi} = U\ket{b} =
\sum_{j=1}^N \beta_j \ket{u_j} \ket{h(\lambda_j)},\]
while the actual algorithm produces the state
\[\ket{\tilde\varphi} = \tilde U \ket{b} =
\tilde P^\dag
\sum_{j=1}^N \beta_j \ket{u_j} \sum_k \alpha_{k|j} \ket{k} \ket{h(\tl_k)},\]
We wish to calculate  $\braket{\tilde \varphi}{\varphi}$,
or equivalently the
inner product between $\tilde P\ket{\tilde\varphi}$ and $\tilde P\ket{\varphi}
= \sum_{j,k} \beta_j \alpha_{k|j} \ket{u_j}\ket{k}\ket{h(\lambda_j)}$.
This inner product is
\[\braket{\tilde\varphi}{\varphi}=
 \sum_{j=1}^N |\beta_j|^2 \sum_k |\alpha_{k|j}|^2
\braket{h (\tl_k)}{h(\lambda_j)}
:=
\bbE_j \bbE_k \braket{h (\tl_k)}{h(\lambda_j)},\]
where we think of $j$ and $k$ as random variables with joint distribution $\text{Pr}(j,k) = |\beta_j|^2 |\alpha_{k|j}|^2$.
Thus
$$
\Real \braket{\tilde\varphi}{\varphi} =
\bbE_j \bbE_k \Real\braket{h(\tl_k)}{h(\lambda_j)}.$$

Let $\delta = \lambda_j t_0 - 2\pi k = t_0(\lambda_j - \tl_k)$.  From
\lemref{no-PS}, $\Real
\braket{h (\tl_k)}{h(\lambda_j)} \geq 1 - c^2\kappa^2\delta^2 / 2t_0^2$, where $c \le \frac{\pi}{2}$ is a constant.
There are two sources of infidelity.  For $\delta \leq 2\pi$, the
inner product is at least $1 - 2\pi^2c^2\kappa^2/t_0^2$.  For larger values
of $\delta$, we use the bound $|\alpha_{k|j}|^2 \leq 64\pi^2 /
(\lambda_j t_0 - 2\pi k)^4$ (proved in \secref{phase-est}) to find an
infidelity contribution that is
$$\leq 2\sum_{k=\frac{\lambda_jt_0}{2\pi} + 1}^\infty \frac{64\pi^2}{\delta^4}
 \frac{c^2\kappa^2\delta^2}{2t_0^2}
 = \frac{64\pi^2c^2\kappa^2}{t_0^2} \sum_{k=1}^\infty
\frac{1}{4\pi^2k^2} = \frac{8\pi^2c^2}{3} \cdot \frac{\kappa^2}{t_0^2}.$$

Summarizing, we find that
$ \Real \braket{\tilde\varphi}{\varphi}
\geq 1 - 5\pi^2c^2\kappa^2/ t_0^{2}$, which translates into
$\|\,\ket{\tilde\varphi} - \ket{\varphi}\,\| \leq 4\pi c\kappa/t_0 =
2\pi^2\kappa / t_0$.  Since the initial state $\ket{b}$ was arbitrary,
this bounds the operator distance $\|\tilde U-U\|$ as claimed in \eq{no-ps-error}.

Turning now to the post-selected case, we observe that
\begin{eqnarray} \ket{x} & := &
\frac{f(A)\ket{b}\ket{\text{well}} + g(A)\ket{b}\ket{\text{ill}}}
{\sqrt{\bra{b}(f(A)^2 + g(A)^2)\ket{b}}}
\\ & = & \frac{\sum_j \beta_j \ket{u_j} (f(\lambda_j)\ket{\text{well}} +
g(\lambda_j)\ket{\text{ill}})}
{\sqrt{\sum_j |\beta_j|^2(f(\lambda_j)^2 + g(\lambda_j)^2)}}
\\ &=:&
\frac{\sum_j \beta_j \ket{u_j} (f(\lambda_j)\ket{\text{well}} +
g(\lambda_j)\ket{\text{ill}})}
{\sqrt{p}}.
\label{eq:x-def-exact}\end{eqnarray}
Where in the last step we have defined
$$p := \bbE_j[f(\lambda_j)^2 + g(\lambda_j)^2] $$
to be the probability that the post-selection succeeds.
Naively, this post-selection could magnify the errors by as much as $1/\sqrt{p}$, but by careful examination of the errors, we find that this worst-case situation only occurs when the errors are small in the first place.  This is what will allow us to obtain the same $O(\kappa/t_0)$ error bound even in the post-selected state.

Now write the actual state that we produce as
\begin{eqnarray}
\ket{\tilde x} & := &
\frac{\tilde P^\dag
\sum_{j=1}^N \beta_j \ket{u_j} \sum_k \alpha_{k|j} \ket{k} (f(\tl_k)\ket{\text{well}} + g(\tl_k)\ket{\text{ill}})}
{\sqrt{\bbE_{j,k} f(\tl_k)^2 + g(\tl_k)^2}}
\\ & =: & \frac{\tilde P^\dag
\sum_{j=1}^N \beta_j \ket{u_j} \sum_k \alpha_{k|j} \ket{k} (f(\tl_k)\ket{\text{well}} + g(\tl_k)\ket{\text{ill}})}
{\sqrt{\tilde p}},
\label{eq:tx-def-exact}\end{eqnarray}
where we have defined $\tilde p = \bbE_{j,k} [f(\tl_k)^2 + g(\tl_k)^2]$.

Recall that  $j$ and $k$ are random variables with joint
distribution $\text{Pr}(j,k) = |\beta_j|^2 |\alpha_{k|j}|^2$.  We evaluate the contribution of a single $j$ value. Define $\lambda := \lambda_j$ and $\tilde{\lambda} := 2\pi k / t_0$.
Note that $\delta = t_0(\lambda - \tilde{\lambda})$ and that $\bbE
\delta, \bbE \delta^2 = O(1)$.  Here $\delta$ depends implicitly on
both $j$ and $k$, and the above bounds on its expectations hold even
when conditioning on an arbitrary value of $j$.  We further abbreviate
$f:=f(\lambda)$, $\tilde f := \tilde f(\lambda)$, $g:=g(\lambda)$ and $\tilde g = \tilde g(\lambda)$.
Thus $p := \bbE [f^2 + g^2]$ and $\tilde{p}= \bbE [\tilde f^2 + \tilde g^2]$.

Our goal is to bound $\| \ket{x} - \ket{\tilde x}\|$ in \eq{fg-ps-error}.
We work instead with the fidelity
\begin{eqnarray}
F &:=& \braket{\tilde x}{x} = \frac{\bbE [f\tilde f + g \tilde g]}{\sqrt{p\tilde{p}}}
=  \frac{\bbE [f^2 + g^2] + \bbE[(\tilde f-f)f + (\tilde g-g)g ]}{p\sqrt{1+\frac{\tilde{p}-p}{p}}}
  \\ &=&
 \frac{1 + \frac{\bbE[(\tilde f-f)f + (\tilde g - g)g]}{p}}{\sqrt{1+\frac{\tilde{p}-p}{p}}}
\geq \l(1 + \frac{\bbE[(\tilde f-f)f + (\tilde g - g)g]}{p}\r)
\l(1-\frac{1}{2}\cdot\frac{\tilde{p}-p}{p}\r)
\label{eq:ps-fid-fg}
\end{eqnarray}
Next we expand
\ba\tilde p -p &= \bbE [\tilde f^2 - f^2] + \bbE [\tilde g^2 - g^2]
\\ &= \bbE[(\tilde f - f)(\tilde f + f)] + \bbE[(\tilde g - g)(\tilde g + g)]
\\ &= 2\bbE[(\tilde f - f)f] + 2\bbE[(\tilde g - g)g]
+ \bbE[(\tilde f-f)^2] + \bbE[(\tilde g-g)^2]
\label{eq:pp-expansion-fg}\ea
Substituting into \eq{ps-fid-fg}, we find
\be F \geq 1 - \frac{\bbE[(\tilde f-f)^2 + (\tilde g-g)^2]}{2p} -
\frac{\bbE[(\tilde f-f)f + (\tilde g-g)g]}{p}\cdot
\frac{\tilde p - p}{2p}
\label{eq:ps-fid-fg2}\ee

We now need an analogue of the Lipschitz condition given in
\lemref{no-PS}.
\begin{lem}\label{lem:fg-cont}
Let $f,\tilde f,g,\tilde g$ be defined as above, with $\kappa^{\prime} = 2 \kappa$. Then
\[|f - \tilde f|^2 + |g - \tilde g|^2 \le c\frac{\kappa^2}{t_0^2} \delta^2 |f^2 + g^2|\]
where $c = \pi^2/2$.
\end{lem}

\begin{proof}
Remember that $\tilde f = f(\lambda - \delta/t_0)$ and similarly for
$\tilde g$.

Consider the case first when $\lambda \geq 1/\kappa$.  In this case $g=0$, and we need to show that
\be |f-\tilde{f}| \leq 2\frac{\kappa |\delta f|}{t_0} =
\frac{|\lambda-\tilde{\lambda}|}{\lambda}
\label{eq:pseudo-lip}\ee
To prove this, we consider four cases.  First, suppose $\tl \geq
1/\kappa$.  Then $|f-\tilde f| = \frac{1}{2\kappa}\frac{|\tl -
  \lambda|}{\tl \cdot \lambda} \leq |\delta|/2t_0\lambda$.  Next, suppose
$\lambda=1/\kappa$ (so $f=1/2$) and $\tl < 1/\kappa$.  Since $\sin
\frac{\pi}{2}\alpha \geq \alpha$ for $0\leq \alpha\leq 1$, we have
\be |f-\tilde f| \leq \frac{1}{2} - \frac{1}{2}\frac{\tl -
  \frac{1}{\kappa'}}{\frac{1}{\kappa} - \frac{1}{\kappa'}}
 = \frac{1}{2} - \kappa(\tl - \frac{1}{2}) =
 \kappa(\frac{1}{\kappa}-\tl),
\label{eq:ftf-diff}\ee
and using $\lambda=1/\kappa$ we find that $|f-\tilde f|
=\frac{\lambda-\tl}{\lambda}$, as desired.
Next, if $\tl < 1/\kappa <\lambda$ and $f<\tilde f$
then replacing $\lambda$ with $1/\kappa$ only makes the inequality
tighter.   Finally, suppose $\tl < 1/\kappa <\lambda$ and $\tilde f <
f$.  Using \eq{ftf-diff} and $\lambda > 1/\kappa$ we find that
$f-\tilde f \leq 1 - \kappa\tl < 1 - \tl/\lambda =
(\lambda-\tl)/\lambda$, as desired.

Now, suppose that $\lambda<1/\kappa$.  Then
\[|f - \tilde f|^2 \leq  \frac{\delta^2}{t_0^2} \max|f^{\prime}|^2  =
\frac{\pi^2}{4}
\frac{\delta^2}{t_0^2} \kappa^2. \]
And similarly
\[|g - \tilde g|^2 \leq \frac{\delta^2}{t_0^2} \max|g^{\prime}|^2
= \frac{\pi^2}{4}\frac{\delta^2}{t_0^2} \kappa^2. \]
Finally $f(\lambda)^2 + g(\lambda)^2 = 1/2$ for any $\lambda\leq 1/\kappa$, implying the result.
\end{proof}

Now we use \lemref{fg-cont} to bound the two error contributions in \eq{ps-fid-fg}.
First bound
\be \frac{\bbE[(\tilde f-f)^2 + (\tilde g -g)^2]}{2p}
\leq O\l(\frac{\kappa^2}{t_0^2}\r)\cdot
\frac{\bbE[(f^2 + g^2)\delta^2]}{\bbE[f^2 + g^2]}
\leq O\l(\frac{\kappa^2}{t_0^2}\r)
\label{eq:fg-bound1}
\ee
The first inequality used \lemref{fg-cont} and the second used the fact
that $\bbE[\delta^2]\leq O(1)$ even when conditioned on an arbitrary
value of $j$ (or equivalently $\lambda_j$).

Next,
\ba \frac{\bbE[(\tilde f-f)f + (\tilde g-g)g]}{p} 
& \leq
\frac{\bbE\l[\sqrt{\l((\tilde f-f)^2 + (\tilde g-g)^2\r)(f^2+g^2)}\r]}{p}
  &\leq
\frac{\bbE\l[\sqrt{\frac{\delta^2\kappa^2}{t_0^2}(f^2+g^2)^2}\r]}{p} 
 &\leq O\l(\frac{\kappa}{ t_0}\r),
\label{eq:fg-bound2}
\ea
where the first inequality is Cauchy-Schwartz, the second is \lemref{fg-cont} and the last uses the fact that $\bbE[|\delta|] \leq \sqrt{\bbE[\delta^2]} = O(1)$ even when conditioned on $j$.

We now substitute \eq{fg-bound1} and \eq{fg-bound2} into \eq{pp-expansion-fg} (and assume $\kappa\leq t_0$) to find
\be \frac{|\tilde p - p|}{p} \leq  O\l(\frac{\kappa}{ t_0}\r).
\label{eq:fg-bound3}\ee
Substituting \eq{fg-bound1}, \eq{fg-bound2} and \eq{fg-bound3} into \eq{ps-fid-fg2}, we find $\Real\braket{\tilde x}{x} \geq 1 -
O(\kappa^2/t_0^2)$, or equivalently, that $\| \ket{\tilde x} -
\ket{x}\| \leq \eps$.  This completes the proof of \thmref{error}.\qed

\subsection{Phase estimation calculations}
\label{sec:phase-est}

Here we describe, in our
notation, the improved phase-estimation procedure of \cite{LP96,BDM98}, and prove the concentration bounds on $|\alpha_{k|j}|$.
Adjoin the state
$$\ket{\Psi_0}=\sqrt{\frac{2}{T}} \sum_{\tau=0}^{T-1}
\sin\frac{\pi(\tau+\frac{1}{2})}{T} \ket{\tau}.$$
Apply the conditional Hamiltonian evolution $\sum_\tau \proj{\tau} \ot
e^{iA\tau t_0 / T}$.  Assume the target state is $\ket{u_j}$, so this
becomes simply the conditional phase
$\sum_\tau \proj{\tau} e^{i\lambda_jt_0\tau/T}$.  The resulting state is
$$\ket{\Psi_{\lambda_jt_0}} = \sqrt{\frac{2}{T}} \sum_{\tau=0}^{T-1} e^{\frac{i\lambda_jt_0\tau}{T}}
\sin\frac{\pi(\tau+\frac{1}{2})}{T} \ket{\tau} \ket{u_j}.$$
We now measure in the Fourier basis, and find that the inner product
with $\frac{1}{\sqrt{T}}\sum_{\tau=0}^{T-1} e^{\frac{2\pi i k\tau}{T}}\ket{\tau}\ket{u_j}$ is (defining $\delta := \lambda_j t_0 - 2\pi k$):
\begin{align}
\alpha_{k|j} & =
\frac{\sqrt{2}}{T} \sum_{\tau=0}^{T-1} e^{i\frac{\tau}{T}(\lambda_j
  t_0 - 2\pi k)} \sin\frac{\pi(\tau+\frac{1}{2})}{T}
\\& =
\frac{1}{i\sqrt{2}T} \sum_{\tau=0}^{T-1}
 e^{i\frac{\tau\delta}{T}}
\l( e^{\frac{i\pi(\tau+1/2)}{T}} - e^{-\frac{i\pi(\tau+1/2)}{T}}\r)
\\ &  =
\frac{1}{i\sqrt{2}T} \sum_{\tau=0}^{T-1}
e^{\frac{i\pi}{2T}} e^{i\tau\frac{\delta + \pi}{T}} -
e^{-\frac{i\pi}{2T}} e^{i\tau\frac{\delta - \pi}{T}}
\\ & =
\frac{1}{i\sqrt{2}T} \l(
e^{\frac{i\pi}{2T}} \frac{1 - e^{i\pi + i \delta}}
{1 - e^{i\frac{\delta + \pi}{T}}} -
e^{-\frac{i\pi}{2T}} \frac{1 - e^{i\pi + i \delta}}
{1 - e^{i\frac{\delta - \pi}{T}}}\r)
\\ & =
\frac{1 + e^{i\delta}}{i\sqrt{2}T}
\l( \frac{e^{-i\delta/2T}}{e^{-\frac{i}{2T}(\delta+\pi)}
- e^{\frac{i}{2T}(\delta+\pi)}} 
- \frac{e^{-i\delta/2T}}{e^{-\frac{i}{2T}(\delta-\pi)}
- e^{\frac{i}{2T}(\delta-\pi)}}\r)
\\ & =
\frac{(1 + e^{i\delta})e^{-i\delta/2T}}{i\sqrt{2}T}
\l( \frac{1}{-2i\sin\l(\frac{\delta+\pi}{2T}\r)} -
\frac{1}{-2i\sin\l(\frac{\delta-\pi}{2T}\r)}\r)
\\ & =
- e^{i\frac{\delta}{2}(1-\frac{1}{T})}
\frac{\sqrt{2}\cos(\frac{\delta}{2})}{ T}
\l( \frac{1}{\sin\l(\frac{\delta+\pi}{2T}\r)} -
\frac{1}{\sin\l(\frac{\delta-\pi}{2T}\r)}\r)
\\ & =
- e^{i\frac{\delta}{2}(1-\frac{1}{T})}
\frac{\sqrt{2}\cos(\frac{\delta}{2})}{ T}
\cdot \frac{\sin\l(\frac{\delta-\pi}{2T}\r) - \sin\l(\frac{\delta+\pi}{2T}\r)}
{\sin\l(\frac{\delta+\pi}{2T}\r)\sin\l(\frac{\delta-\pi}{2T}\r)}
\\ & =
e^{i\frac{\delta}{2}(1-\frac{1}{T})}
\frac{\sqrt{2}\cos(\frac{\delta}{2})}{ T}
\cdot \frac{2\cos\l(\frac{\delta}{2T}\r)\sin\l(\frac{\pi}{2T}\r)}
{\sin\l(\frac{\delta+\pi}{2T}\r)\sin\l(\frac{\delta-\pi}{2T}\r)}
\end{align}
Following \cite{LP96,BDM98}, we make the assumption that $2\pi \leq \delta
\leq T/10$.  Further using $\alpha-\alpha^3/6 \leq \sin \alpha \leq \alpha$ and ignoring phases we find that
\ba|\alpha_{k|j}| & \leq
\frac{4\pi\sqrt{2}}{(\delta^2-\pi^2)(1-\frac{\delta^2+\pi^2}{3T^2})} \leq \frac{8\pi}{\delta^2}.\ea
Thus $|\alpha_{k|j}|^2 \leq 64\pi^2 / \delta^2$
whenever $|k - \lambda_j t_0 / 2\pi|\geq 1$.

\subsection{The non-Hermitian case}
\label{sec:non-hermitian}
Suppose $A\in\bbC^{M\times N}$ with $M\leq N$.  Generically $Ax=b$ is now
underconstrained.
Let the singular value decomposition of $A$ be
$$A = \sum_{j=1}^M \lambda_j \ket{u_j}\bra{v_j},$$
with $\ket{u_j}\in\bbC^M$, $\ket{v_j}\in\bbC^N$ and $\lambda_1\geq
\cdots \lambda_M\geq 0$.  Let $V=\text{span}\{\ket{v_1},\ldots,\ket{v_M}\}$.
Define
\be H =\begin{pmatrix}
  0     & A \\
  A^{\dagger} & 0
\end{pmatrix}\label{eq:H-def-app}.\ee
  $H$ is Hermitian with eigenvalues $\pm
\lambda_1,\ldots,\pm \lambda_M$, corresponding to eigenvectors
$\ket{w^{\pm}_j}:=\frac{1}{\sqrt{2}}(\ket{0}\ket{u_j} \pm
\ket{1}\ket{v_j})$.  It also
has $N-M$ zero eigenvalues, corresponding to the orthogonal complement
of $V$.

To run our algorithm we use the input $\ket{0}\ket{b}$.  If $\ket{b} =
\sum_{j=1}^M \beta_j \ket{u_j}$ then $$\ket{0}\ket{b}=\sum_{j=1}^M
\beta_j\frac{1}{\sqrt{2}}(\ket{w^+_j} + \ket{w^-_j})$$
 and running the inversion algorithm yields a state proportional to
$$H^{-1}\ket{0}\ket{b} =
\sum_{j=1}^M
\beta_j\lambda_j^{-1}\frac{1}{\sqrt{2}}(\ket{w^+_j} - \ket{w^-_j})
=\sum_{j=1}^M
\beta_j\lambda_j^{-1}\ket{1}\ket{v_j}.$$
Dropping the inital $\ket{1}$, this defines our solution $\ket{x}$.
Note that our algorithm does not produce any component in $V^\perp$,
although doing so would have also yielded valid solutions.  In this
sense, it could be said to be finding the $\ket{x}$ that minimizes
$\braket{x}{x}$ while solving $A\ket{x}=\ket{b}$.

On the other hand, if $M\geq N$ then the problem is overconstrained.
Let $U=\text{span}\{\ket{u_1},\ldots,\ket{u_N}\}$.  The equation
$A\ket{x}=\ket{b}$ is satisfiable only if $\ket{b}\in U$.  In this
case, applying $H$ to $\ket{0}\ket{b}$ will return a valid solution.
But if $\ket{b}$ has some weight in $U^\perp$, then $\ket{0}\ket{b}$
will have some weight in the zero eigenspace of $H$, which will be
flagged as ill-conditioned by  our algorithm.  We might choose to
ignore this part, in which case the algorithm will return an $\ket{x}$
satisfying $A\ket{x}  = \sum_{j=1}^N \proj{u_j} \,\ket{b}$.

\subsection{Optimality}
\label{sec:reduction}

In this section, we explain in detail two important ways in which our
algorithm is optimal up to polynomial factors.  First, no classical
algorithm can perform the same matrix inversion task; and second, our
dependence on condition number and accuracy cannot be substantially
improved.

We present two versions of our lower bounds; one based on
complexity theory, and one based on oracles.  We say that an algorithm solves matrix inversion if its input and output are
\begin{enumerate} \item Input: An $O(1)$-sparse matrix $A$ specified either via an oracle or via a $\poly(\log(N))$-time algorithm that returns the nonzero elements
in a row.
\item Output:  A bit that equals one with probability
$\bra{x}M\ket{x}\pm \eps$, where $M=\proj{0}\otimes I_{N/2}$
corresponds to measuring the first qubit and $\ket{x}$ is a normalized
state proportional to $A^{-1}\ket{b}$ for $\ket{b} = \ket{0}$.
\end{enumerate}
Further we demand that $A$ is
Hermitian and $\kappa^{-1}I \leq A \leq I$.  We take $\eps$ to be a
fixed constant, such as 1/100, and later deal with the dependency in $\epsilon$.  If the algorithm works when $A$ is specified by an oracle, we say that it is relativizing.   Even though this is a very weak definition of inverting matrices, this task is still hard for classical computers.

\begin{thm}\label{thm:no-go}
\benum
\item If a quantum algorithm exists for matrix inversion running in
  time $\kappa^{1-\delta}\cdot \poly\log(N)$ for some
  $\delta>0$, then {\bf
    BQP}={\bf PSPACE}.
\item No relativizing quantum algorithm can run in time
$\kappa^{1-\delta}\cdot \poly\log(N)$.
\item If a classical algorithm exists for matrix inversion running in
  time $\poly(\kappa,\log(N))$, then {\bf
    BPP}={\bf BQP}.
\eenum
\end{thm}

Given an $n$-qubit $T$-gate quantum computation, define $U$ as in
\eq{feynman-unitary}.  Define
\be A=\begin{pmatrix} 0 & I-Ue^{-\frac{1}{T}} \\
I-U^\dag e^{-\frac{1}{T}} & 0 \end{pmatrix}.\ee
Note that $A$ is Hermitian, has condition number $\kappa\leq 2T$ and
dimension $N=6T2^n$.
Solving the matrix inversion
problem corresponding to $A$ produces an $\eps$-approximation of the
quantum computation corresponding to applying $U_1,\ldots,U_T$, assuming we are allowed to make any two outcome measurement on the output state $\ket{x}$.
Recall that
\be  \left( I-Ue^{-\frac{1}{T}} \right)^{-1} = \sum_{k\geq 0} U^k e^{-k/T} .\ee
We define a measurement $M_0$, which outputs zero if the time register
$t$ is between $T+1$ and $2T$, and the original measurement's output
was one. As  $ \Pr(T+1 \le k \le 2T) =
e^{-2}/(1 + e^{-2} + e^{-4})$
and is independent of the result of the measurement $M$, we can
estimate the expectation of $M$ with accuracy $\eps$ by iterating this
procedure $O \left( 1/\eps^2 \right)$ times.

In order to perform the simulation when measuring only the first
qubit, define \be B=\begin{pmatrix} I_{6 T 2^n} & 0 \\
0 & I_{3T 2^n} - Ue^{-\frac{1}{T}} \end{pmatrix}.\ee
We now define $\tilde B$ to be the matrix $B$, after we permuted the rows and columns such that if
\be C=\begin{pmatrix} 0 & \tilde B \\
\tilde B^{\dag} & 0 \end{pmatrix}.\ee
and $C\vec y = \begin{pmatrix}\vec b \\ 0 \end{pmatrix}$, then
measuring the first qubit of $\ket{y}$ would correspond to perform
$M_0$ on $\ket{x}$. The condition number of $C$ is equal to that of
$A$, but the dimension is now $N = 18 T 2^n$.

Now suppose we could solve matrix inversion in time
$\kappa^{1-\delta}(\log(N)/\eps)^{c_1}$ for constants
$c_1\geq 2,\delta>0$.  Given a computation with $T\leq 2^{2n}/18$, let
$m=\frac{2}{\delta}\frac{\log(2 n)}{\log(\log(n))}$ and $\eps
= 1/100m$.  For sufficiently large $n$, $\eps \geq 1/\log(n)$.
Then
$$\kappa^{1-\delta}\l(\frac{\log(N)}{\eps}\r)^{c_1} \leq
(2T)^{1-\delta} \l(\frac{3n}{\eps}\r)^{c_1} \leq
T^{1-\delta} c_2 (n\log(n))^{c_1},$$
where $c_2 = 2^{1- \delta}3^{c_1}$ is another constant.

We now have a recipe for simulating an
$n_i$-qubit $T_i$-gate computation with $n_{i+1}=n_i+\log(18T_i)$ qubits,
$T_{i+1}=T_i^{1-\delta}c_3(n_i\log(n_i))^{c_1}$ gates and error $\eps$.  Our
strategy is to start with an $n_0$-qubit $T_0$-gate computation and
iterate this simulation $\ell\leq m$ times, ending with an $n_\ell$-qubit
$T_\ell$-gate computation with error $\leq m\eps \leq 1/100$. We stop
iterating either after $m$ steps, or whenever
$T_{i+1}>T_i^{1-\delta/2}$, whichever
comes first. In the latter case, we set $\ell$ equal to the first $i$
for which $T_{i+1}>T_i^{1-\delta/2}$.

In the case where we iterated the reduction $m$ times, we have $T_i
\leq T^{(1-\delta/2)^i} \leq 2^{(1-\delta/2)^i 2 n_0}$, implying that
$T_m \leq n_0$.  On the other hand, suppose we stop for some $\ell<m$.
For each $i< \ell$ we have $T_{i+1}\leq T_i^{1-\delta/2}$.  Thus $T_i
\leq 2^{(1-\delta/2)^i2n_0}$ for each $i\leq \ell$.  This allows us to
bound $ n_i = n_0 + \sum_{j=0}^{i-1} \log(18 T_i) = n_0 +
2n_0\sum_{j=0}^{i-1} (1-\delta/2)^j + i\log(18) \leq \l(
\frac{4}{\delta} + 1\r) n_0 + m\log(18).$ Defining yet another
constant, this implies that $T_{i+1} \leq T_i^{1-\delta}
c_3(n_0\log(n_0))^{c_1}$.  Combining this with our stopping condition
$T_{\ell+1}>T_\ell^{1-\delta/2}$ we find that
$$T_\ell \leq \l(c_3(n_0\log(n_0))^{c_1}\r)^{\frac{2}{\delta}} =
\poly(n_0).$$
Therefore, the runtime of the procedure is polynomial in $n_0$
regardless of the reason we stopped iterating the procedure.  The
number of qubits used increases only linearly.

Recall that the TQBF (totally quantified Boolean formula
satisfiability) problem is {\bf PSPACE}-complete, meaning that any
$k$-bit problem instance for any language in {\bf PSPACE} can be
reduced to a TQBF problem of length $n=\poly(k)$ (see \cite{complexityBook} for more information).  The formula can be
solved in time $T \leq 2^{2n}/18$, by exhaustive enumeration over the
variables. Thus a {\bf PSPACE} computation can be solved in quantum
polynomial
time.  This proves the first part of the theorem.

To incorporate oracles, note that our construction of $U$ in
\eq{feynman-unitary} could simply replace some of the $U_i$'s with
oracle queries.  This preserves sparsity, although we need the rows of
$A$ to now be specified by oracle queries.  We can now iterate the
speedup in exactly the same manner.  However, we conclude with the
ability to solve
the OR problem on $2^n$ inputs in $\poly(n)$ time and queries.  This,
of course, is impossible \cite{BBBV}, and so the purported relativizing
quantum algorithm must also be impossible.

The proof of part 3 of \thmref{no-go} simply formulates a
$\poly(n)$-time, $n$-qubit quantum computation as a $\kappa=\poly(n)$,
$N=2^n\cdot\poly(n)$ matrix inversion problem and applies the
classical algorithm which we have assumed exists.
\qed

Theorem \ref{thm:no-go} established the universality of the matrix inversion algorithm. To extend the simulation to problems which are not decision problems, note that the algorithm actually supplies us with $\ket{x}$ (up to some accuracy). For example, instead of measuring an observable $M$, we can measure $\ket{x}$ in the computational basis, obtaining the result $i$ with probability $|\braket{i}{x}|^2$. This gives a way to simulate quantum computation by classical matrix inversion algorithms. In turn, this can be used to prove lower bounds on classical matrix inversion algorithms, where we assume that the classical algorithms output samples according to this distribution.

\begin{thm}No relativizing classical matrix inversion algorithm can
  run in time $N^{\alpha}2^{\beta\kappa}$ unless $3\alpha +
  4\beta \geq 1/2$.
\end{thm}

If we consider matrix inversion algorithms that work only on positive
definite matrices, then the $N^\alpha 2^{\beta\kappa}$ bound becomes
$N^\alpha 2^{\beta\sqrt{\kappa}}$.

\begin{proof}
Recall Simon's problem \cite{simon}, in which we are given $f:\bbZ_2^n\ra
\{0,1\}^{2n}$ such that $f(x)=f(y)$ iff $x+y=a$ for some
$a\in\bbZ_2^n$ that we would like to find.  It can be solved by
running a $3n$-qubit $2n+1$-gate quantum computation $O(n)$ times and
performing a $\poly(n)$ classical
computation. The randomized classical lower bound is
$\Omega(2^{n/2})$ from birthday arguments.

Converting Simon's algorithm to a matrix $A$ yields $\kappa \approx
4n$ and $N\approx 36n2^{3n}$.  The run-time is
 $N^{\alpha}2^{\beta{\kappa}} \approx
2^{(3\alpha+4\beta)n}\cdot\poly(n)$.  To avoid violating the oracle
lower bound, we must have $3\alpha + 4\beta \geq 1/2$, as required.
\end{proof}

Next, we argue that the accuracy of algorithm cannot be substantially
improved.  Returning now to the problem of estimating
$\bra{x}M\ket{x}$, we recall that classical algorithms can approximate
this to accuracy $\eps$ in time $O(N{\kappa}\poly(\log(1/\eps)))$.
This $\poly(\log(1/\eps))$ dependence is because when writing the
vectors $\ket{b}$ and $\ket{x}$ as bit strings means that adding an
additional bit will double the accuracy.  However, sampling-based
algorithms such as ours cannot hope for a better than $\poly(1/\eps)$
dependence of the run-time on the error.   Thus proving that our
algorithm's error performance cannot be improved will require a slight
redefinition of the problem.

Define the matrix inversion estimation problem as follows.  Given $A, b, M,\eps,\kappa,s$ with $\|A\|\leq 1, \|A^{-1}\|\leq \kappa$, $A$ $s$-sparse and efficiently row-computable, $\ket{b}=\ket{0}$ and $M=\proj{0}\otimes I_{N/2}$: output a number that is within $\eps$ of $\bra{x}M\ket{x}$ with probability $\geq 2/3$, where $\ket{x}$ is the unit vector proportional to $A^{-1}\ket{b}$.

The algorithm presented in our paper can be used to solve this problem
with a small amount of overhead.  By producing $\ket{x}$ up to trace
distance $\eps/2$ in time $\tilde O(\log(N)\kappa^2s^2/\eps)$, we can
obtain a sample of a bit which equals one with probability $\mu$ with
$|\mu - \bra{x}M\ket{x}|\leq \eps/2$.  Since the variance of this bit
is $\leq 1/4$, taking $1/3\eps^2$ samples gives us a $\geq 2/3$
probability of obtaining an estimate within $\eps/2$ of $\mu$.   Thus
quantum computers can solve the matrix inversion estimation problem in
time $\tilde O(\log(N)\kappa^2s^2/\eps^3)$.

We can now show that the error dependence of our algorithm cannot be substantially improved.
\begin{thm}\label{thm:error-dep-no-go}
\benum \item
If a quantum algorithm exists for the matrix inversion estimation problem running in time $\poly(\kappa, \log(N), \log(1/\eps))$ then {\bf BQP}={\bf PP}.
\item No relativizing quantum algorithm for the matrix inversion estimation problem can run in time $N^{\alpha}\poly(\kappa)/\eps^\beta$ unless $\alpha + \beta  \geq 1$.
\eenum
\end{thm}

\begin{proof}
\benum\item
A complete problem for the class {\bf PP} is to count the number of satisfying assignments to a SAT formula. Given such formula $\phi$, a quantum circuit can apply it on a superposition of all $2^n$ assignments for variables, generating the state
\[\sum_{z_1, \ldots, z_n\in\{0,1\}} \ket{z_1, \ldots, z_n} \ket{\phi(z_1, \ldots z_n)}.\]
The probability of obtaining $1$ when measuring the last qubit is equal to the number of satisfying truth assignments divided by $2^n$. A matrix inversion estimation procedure which runs in time $\poly\log(1/\epsilon)$ would enable us to estimate this probability to accuracy $2^{-2n}$ in time $\poly(\log (2^{2n})) = \poly(n)$. This would imply that {\bf BQP} = {\bf PP} as required.
\item Now assume that $\phi(z)$ is provided by the output of an
  oracle.  Let $C$ denote the number of $z\in\{0,1\}^n$ such that
  $\phi(z)=1$.  From \cite{parity}, we know that determining the parity
  of $C$ requires $\Omega(2^{n})$ queries to $\phi$.
  However, exactly determining $C$ reduces to the matrix inversion
  estimation problem with $N=2^n$, $\kappa=O(n^2)$ and
  $\eps=2^{-n-2}$.  By assumption we can solve this in time $2^{(\alpha
    + \beta)n}\cdot\poly(n)$, implying that $\alpha+\beta\geq 1$.
\eenum
\end{proof}

\end{document}